\documentclass[a4paper,twoside,11pt]{article}
\usepackage{amssymb}
\usepackage{amstext}
\usepackage{amsmath}
\usepackage{amsthm}
\usepackage{multirow}
\newtheorem{definition}{Definition}[section]
\newtheorem{theorem}{Theorem}[section]
\newtheorem{lemma}[theorem]{Lemma}

\usepackage{fullpage}
\newcommand\Tstrut{\rule{0pt}{2.6ex}}       
\newcommand\Bstrut{\rule[-0.9ex]{0pt}{0pt}} 
\newcommand{\TBstrut}{\Tstrut\Bstrut} 
\newcommand{\etal}{\textit{et al}.} 
\newcommand{\ie}{\textit{i}.\textit{e}.}

\newcommand{\N}{Niederreiter }

\newcommand{\Addresses}{{
  \bigskip
  \footnotesize

  \textsc{IISER Pune, Pune, India}\par\nopagebreak
  \textit{E-mail address}, U.~Kapshikar: \texttt{uskapshikar@gmail.com}\par\nopagebreak
  \textit{E-mail address}, A.~Mahalanobis: \texttt{ayan.mahalanobis@gmail.com}}}

\title{A Quantum-Secure \N Cryptosystem using Quasi-Cyclic Codes}
\author{Upendra Kapshikar \and Ayan Mahalanobis}
\date{}
\begin{document}
\maketitle
\abstract{In this paper, we describe a new variant of Niederreiter cryptosystem over quasi-cyclic codes of rate $\frac{m-1}{m}$. We show that the proposed cryptosystem is quantum secure, in particular, it resists quantum Fourier sampling and has better transmission rate with smaller keys compared to the one using binary Goppa codes. }

\vspace{0.1cm}

\section{Introduction}
\label{sec:introduction}
\noindent In this paper, we develop a new \N cryptosystem using $\frac{m-1}{m}$ quasi-cyclic codes. These are well known \emph{linear block codes}. The main result in this paper is, the \N cryptosystem built with these codes is \emph{quantum-secure}. Many attempts were made by many distinguished authors to built new McEliece or Niederreiter cryptosystems using different codes. However, none of these authors demonstrated that their cryptosystem is quantum-secure. One of the main reason to study the \N cryptosystem is that it can be quantum-secure. In today's world designing cryptosystem just got twice as hard. One needs to ensure that it is secure from known classical attacks \emph{as well as} secure from the known quantum attacks, \ie, from the hidden subgroup problem arising out of quantum Fourier sampling.

The concept of security which is tied to quantum Fourier sampling for non-abelian groups has a rich tradition in the work of Hallgreen \etal~\cite{hallgreen} and Kempe and Salev~\cite{Kempe}. Dinh \etal~\cite{Dinh} using earlier work showed that the McEliece cryptosystem built on binary Goppa codes is quantum-secure. We use their theorem~\cite[Theorem 4]{Dinh} to show that our proposed cryptosystem is quantum-secure. This work is part of first author's master's thesis~\cite{th}.
\paragraph{Description of the parity check matrix used for the proposed Niederreiter cryptosystem}\label{dd}
Recall that we are talking about $\frac{m-1}{m}$ quasi-cyclic codes over $\mathbb{F}_{2^l}$. For the cryptosystem to be quantum-secure the parity check matrix $\mathcal{H}$ for the $[n=mp,k=(m-1)p]$, $\frac{m-1}{m}$ quasi-cyclic code should satisfy the following conditions:
\begin{itemize}
\item[I] Integers $m,p$, such that $p$ is a prime and $m$ is bounded above by a polynomial in $p$.
\item[II] The matrix $\mathcal{H}$ is of size $p\times mp$ over $\mathbb{F}_{2^l}$.
\item[III] The matrix $\mathcal{H}$ is of the form $\mathcal{H}=\left[\,C_0=I\,|\,C_1\,|\,C_2\,|\ldots\,|\,C_{m-1}\,\right]$, where each $C_i$ is a circulant matrix of size $p$. Each $C_i$ for $i>0$ should contain an element from a proper extension of $\mathbb{F}_2$.
Furthermore, we denote the matrix $\mathcal{H}$ as $\left[\,I\,|\,C\,\right]$ where $C$ is the concatenation of the circulant matrices $C_i$, $i>0$.
\item[IV] Define $T_{\mathcal{H}}$ corresponding to $\mathcal{H}$ as $T_{\mathcal{H}}=\left\{P_1\in \mathrm{S}_p \;|\; \exists P_2\in \mathrm{S}_{p(m-1)} \;\text{s.t.}\; P_1CP_2=C\right\}$, where $\mathrm{S}_n$ is the symmetric group acting on $n$ letters. It is easy to see that $T_\mathcal{H}$ is a permutation group action on $p$ letters. The condition we impose on $\mathcal{H}$ is that $T_\mathcal{H}$ is not 2-transitive.
\item[V] No two columns of $C$ are identical. 
\end{itemize}
For more on quantum security of the proposed cryptosystem a reader can skip to Section~\ref{variant}.
\section{Review of Niederreiter Cryptosystem}
\noindent In 1978, Robert McEliece~\cite{ME} suggested a cryptosystem based on algebraic codes. The system did not gain much popularity then because of its large key sizes. However, with quantum computing becoming a reality, McEliece cryptosystems have become the center of attention for cryptographers. Unlike RSA or ElGamal cryptosystem, McEliece cryptosystem is based on a non-commutative structure which allegedly makes it a strong candidate for \emph{post-quantum cryptography}. Compared to traditional cryptosystems, a McEliece cryptosystem has following advantages:
\begin{description}
\item[a)] It is fast. Faster than RSA or ElGamal.
\item[b)] It is believed to be quantum secure.
\end{description}
The disadvantages are the following:
\begin{description}
\item[a)] The key sizes are huge.
\item[b)] The ciphertext becomes much larger than the plaintext because of the redundancy added by the encoding process.
\end{description} 
In this paper we focus on building a \N cryptosystem. Apart from fulfilling the obvious requirement of a signature scheme, Niederreiter cryptosystem is faster compared to a McEliece cryptosystem. Li \etal~\cite{Li} proved that the security of a McEliece cryptosystem and its Niederreiter counterpart are equivalent under the Lee-Brickell attack~\cite{Lee}. For these reasons, in this paper our main focus is on \textbf{building a Niederreiter cryptosystem}.

All codes in this paper are \emph{linear block codes} and our standard reference for coding theory is Blahut~\cite[Chapter 3]{blahut}.
A binary linear code $\mathcal{C}$ of length $n$ and rank $k$ is a $k$ dimensional linear subspace of $\mathbb{F}_{2}^n$. Hamming weight or simply weight $t$ of a codeword means that the codeword has $t$ non-zero entries. Standard distance on $\mathcal{C}$ is the hamming distance. Distance of a linear code $\mathcal{C}$ is defined as minimum distance between two non-zero codewords. Traditionally, such a code is denoted as  $\left[\,n,k\right]$ code. The ratio $\frac{k}{n}$
is the transmission rate of the code.

A generator matrix $M$ of an $\left[n,k \right]$ linear code $\mathcal{C}$ is a $k\times n$ matrix such that $\mathcal{C}= \lbrace xG: x\in \mathbb{F}_{2}^k \rbrace$. A generator matrix $M$ of the form $\left[\,I_{k}\,|\, G^{\prime} \right]$ is said to be in the systematic form. A parity check matrix $\mathcal{H}$ of an $\left[ n,k \right]$ linear code $\mathcal{C}$ is an $\left( n-k \right) \times n$ such that $\mathcal{H}x=0$ if and only if $x\in \mathcal{C}$. A code generated by $\mathcal{H}$ is known as the dual code and denoted by $\mathcal{C}^\perp.$

A decoding problem is given $x\in \mathbb{F}_{2}^n$  find a codeword $c\in \mathcal{C}$ that is closest to $x$. This problem for a random linear code is called the \emph{general decoding problem} and is known to be NP-hard~\cite{berlekamp1978inherent}. But for some codes this decoding problem can be solved efficiently. If such an algorithm is available for a code $\mathcal{C}$, we say that $\mathcal{C}$ has a decoder.
\subsection{Quasi-Cyclic Codes}
\noindent A quasi-cyclic code (QCC) is a block linear code which is a simple generalisation of the cyclic code. It is such that any cyclic shift of any codeword by $m$ symbols gives another codeword. 
We are particularly interested in $\frac{m-1}{m}$ rate codes. In particular our system is based on $\frac{m-1}{m}$ rate codes over $\mathbb{F}_{{2}^{l}}$. Such codes along with quasi-cyclic codes of rate $\frac{1}{m}$ were studied in great detail by Gulliver~\cite{Gulliver_thesis}. 
\begin{definition}[Circulant matrix] A square matrix is called circulant if every row, except for the first row, is a circular right shift of the row above that.
\end{definition} 
A rate $\frac{m-1}{m}$ systematic QCC has an $p \times mp$ parity check matrix of the form 
$\mathcal{H}= \left[\,I_{p}\,|\,C_{1}\,|\,C_{2}\,|\,\ldots\,|\,C_{m-1}\,\right]$ where each $C_{i}$ is a circulant matrix of size $p$ and $I_{p}$ is the identity matrix of size $p$.

For compactness we denote $\mathcal{H}= \left[\,I\,|\, C\,\right]$. In a recent work Aylaj \textit{\etal}~\cite{Aylaj} found a way to construct generator matrices for such codes over $\mathbb{F}_{2}$. Since these generator matrices are in systematic form one can easily construct parity check matrix from these generator matrices. As said previously, our main interest lies in codes over extension fields. Gulliver~\cite[Chapter 6]{Gulliver_thesis} shows that quasi-cyclic codes over extension fields can be MDS (maximum distance separable) codes. As the name suggests,  MDS codes can achieve large minimum distance and so there would be no low weight codewords. This plays an imporatnt role in the classical security of McEliece and Niederreiter cryptosystems, such as the Stern's attack and the Lee-Brickell attack. Though Gulliver~\cite{Gulliver_thesis} only presents examples of MDS codes with rate $\frac{1}{m}$, he does present a case to study quasi-cyclic codes of rate $\frac{m-1}{m}$ with large minimum distances. 

\subsection{Decoding} \noindent Quasi-cyclic codes are well studied and well established codes and depending on how one constructs them there are various decoders available. We briefly mention some of them here. Gulliver ~\cite[Appendix B]{Gulliver_thesis} presents a ML (majority logic) decodable QCCs. Another new and intersting way of decoding quasi-cyclic codes using Gr\"{o}bner basis can be found in the work of Zeh and Ling~\cite{zeh}.

\subsection{Niederreiter Cryptosystem}
\noindent Let $\mathcal{H}$ be a $k \times n$ parity check matrix for a $[n,n-k]$ linear code $\mathcal{C}$ for which a fast decoding algorithm exists. Let $t$ be the number of errors that $\mathcal{C}$ can correct.  

\noindent\textbf{Private Key}: ($A_0$,$\mathcal{H}$,$B_0$) where $A_0 \in \rm{\rm{GL}}_{k}(\mathbb{F}_{2})$ and  $B_0$ is a permutation matrix of size $n$.

\noindent\textbf{Public Key}: $\mathcal{H}^{\prime}=A_0\mathcal{H}B_0$.

\noindent\textbf{Encryption}: 
\begin{description}
\item Let $\mathcal{X}$ be a $n$-bit plaintext with weight at most $t$. The corresponding ciphertext $\mathcal{Y}$ of $k$-bits is obtained by calculating $\mathcal{Y}=\mathcal{H}^\prime\mathcal{X}^{\rm{T}}$.

\end{description}
\textbf{Decryption}:
\begin{description}
\item Compute $y=A_0^{-1}\mathcal{Y}$. Thus $y = \mathcal{H}B_0\mathcal{X}^{\rm{T}}$.
\item Using Gaussian elimination find a $z$ such that  weight of $z$ is at most $t$ and $\mathcal{H}z^{\rm{T}}=y$. Since $y = \mathcal{H}B_0\mathcal{X}^{\rm{T}}$, $\mathcal{H}\left(z^{\rm{T}}-B_0\mathcal{X}^{\rm{T}}\right)=0$. Hence we have $z-\mathcal{X}B_0^{\rm{T}} \in \mathcal{C}$.
\item Now use fast decoding on $z$ with $\mathcal{H}$ to get $\mathcal{X}B_0^{\rm{T}}$ and thus recover $\mathcal{X}$.
\end{description}
For further details on \N cryptosystem the reader is referred to~\cite[Chapter 6]{NN}.
\section{Classical Attacks}
\noindent In this seciton we briefly go over the generic classical attacks against McEliece and Niederreiter cryptosystems. We also mention some attacks exploiting the circulant structures in the keys. Interestingly, Li \textit{\etal}~\cite[Section III]{Li} proved (see also~\cite[Theorem 6.4.1]{NN}) that both McElice and Niederreiter cryptosystems are equivalent in terms of classical security. The proof follows from the fact that the encryption equation for one can be reduced to the other. This implies the equivalence of security of both the cryptosystems for attacks that try to extract the plaintext from a ciphertext.
  
Most generic attacks over algebraic code based cryptosystems are \emph{information set decoding attacks}(ISD). Two most popular ways of implementing ISD attacks are by Lee and Brickell~\cite{Lee} and Stern~\cite{stern1988method}. As mentioned by Baldi \etal~\cite{Baldi} ISD attacks are the best known attacks with the least work factor as far as classical cryptanalysis is considered. Hence these work factors are considered as security levels for a McEliece and Niederreiter cryptosystems.  

The basic idea behind one of the attacks was suggested by McEliece himself. Lee and Brickell~\cite{Lee} improved the attack and added an important verification step where attacker confirms that recovered message is the correct one. In this case, we are dealing with a McEliece cryptosystem over a $[n,k]$ linear code. The strategy is based on repeatedly selecting $k$ bits at random from a $n$-bit ciphertext in hope that none of the selected bits are part of the error. Similar attacks can also be implemented over Niederreiter cryptosystems.  Lee and Brickwell also provided a closed-form equation for complexity of the attack. As our system is based on $\left( n=mp, k=(m-1)p, d_{min} = 2\mathcal{E}+1\right)$ code the expression for minimal work factor (with $\alpha = \beta = 1$ as taken by Lee and Brickell) takes the following form $$
W_{min}= W_{2}= T_{2}  \left( (m-1)^{3} p  ^{3} + (m-1) p N_{2} \right)$$ where $T_{2}= \dfrac{1}{Q_{0} + Q_{1} + Q_{2}}$ and $Q_{i}=$ $\mathcal{E} \choose i$ ${n-\mathcal{E}} \choose {k-i} $ / $n \choose k$ with $N_{2}= 1 + k + {k \choose 2}$.

In Table~\ref{table} we present numerical data for work factor for diferent values of parameters. Recently, Aylaj \textit{\etal}~\cite{Aylaj} developed an algorithm to construct stack-circulant codes with high error correcting capacity which makes the proposed Niederreiter cryptosystem much more promising.

Other ISD attacks are based on a strategy given by Stern. To recover the intentional error vector $e$ in a McEliece cryptosystem such strategies use an extension code $C^{\prime \prime}$ generated by generator matrix $M^{\prime \prime}=\left[ \begin{array}{c}
M^\prime \\
x 
\end{array} \right]$.  Bernstein \etal~\cite{lange} later improved this attack. Probability of success and work factor for Stern's attack is described in ~\cite{Hirotomo}. In Table~\ref{table} we also provide probability of success for parameters $l=16$ and $A_{w} \approx n-k$. Both the parameters can be optimized further to obtain the least work factor but not much variation is seen as we change any of these parameters. With such low probabilities, it is clear that the work factor for Stern's attack is worse than the Lee-Brickell attack. Even when one considers improvements suggested by Bernstein \textit{\etal}~\cite{lange}, Lee-Brickell's~\cite{Lee} attack seems to outperform the attack by Bernstein \textit{\etal}~as it produces speedup upto 12 times and hence the security of the system against the Lee-Brickell attack should be considered the security of the system. Key sizes should be devised according to that. 

Another attack worth mentioning for quasi-cyclic codes is the attack on the dual code. This attack works only if the dual code has really low weight codewords and is often encountered only when sparse parity check matrices are involved. For example, McEliece with QC-LDPC~\cite{Baldi}. Such attacks can easily be stopped by choosing codes that do not have low weight codewords. From the work of Aylaj \etal~\cite{Aylaj} this can be achieved.

After this discussion on classical security we now move towards quantum-security of the proposed McEliece and Niederreiter cryptosystems which is one of the major goal of this paper.
\section{Quantum Attacks}
\subsection{Hidden Subgroup Problem}
\noindent Quantum Fourier sampling works behind the scene for almost all known quantum algorithms. It is the reason Shor's algorithm for factoring and solving the discrete logarithm problem works. The main idea in this paper is to show that quantum Fourier sampling will not work in some situations, in particular, a \N cryptosystem using quasi-linear $\frac{m-1}{m}$ codes.

However, before we go there let us briefly sketch how quantum Fourier sampling is used to break a McEliece or \N cryptosystems. We recall, the scrambler-permutation attack~\cite[Section 2]{Dinh}. This structural attack is exactly same for a McEliece or a \N cryptosystem except that instead of finding a scrambler-permutation pair from generator matrix $M$ to $M^\prime$ one has to find scrambler-permutation pair from parity check matrix $\mathcal{H}$ to $\mathcal{H}^\prime$. The problem essentially remains the same. In this attack, we assume $\mathcal{H}$ and $\mathcal{H}^\prime$ are known, the attack is to find $A$ and $B$. Notice that finding  any $A^\prime$ and $B^\prime$ such that $A^\prime\mathcal{H}B^\prime=\mathcal{H}^\prime$ will also make the attack successful.

\begin{definition}[Hidden Shift Problem] Let $\mathrm{G}$ be a group. Let $f_{0}$ and $f_{1}$ be two functions from group $\mathrm{G}$ to a set $\mathrm{X}$. Given $f_{0}(g)=f_{1}(g_{0}g)$ the task is to find a constant $g_{0} \in \mathrm{G}$. Note that there can be many $g_{0}$ that satisfy the above condition. Hidden shift problem asks us to find any one of those constants.
\end{definition}

Let $M'=AMB$. A \N cryptosystem will be broken if we find one possible pair $(A,B)$ from $M$ and $M'$. Consider two functions from group $\mathrm{G}=\rm{\rm{GL}}_{k}(\mathbb{F}_{2})\times \mathrm{S}_{n} $ given by  
\begin{equation}
f_{0}(A,B)=A^{-1}MB 
\end{equation}
\begin{equation}
f_{1}(A,B)=A^{-1}M'B.
\end{equation}
Then one can check that $f_{1}(A,B)=f_{0}((A_{0}^{-1},B_{0}).(A,B))$, that is $A_{0}^{-1},B_{0}$ is the shift between $f_{0}$ and $f_{1}$. Hence, if one can solve the hidden shift problem over $\mathrm{G}=\rm{\rm{GL}}_{k}(\mathbb{F}_{2}) \times\mathrm{S}_{n}$ he can break the \N cryptosystem.
The general procedure to solve this hidden shift problem is to reduce it to a hidden subgroup problem.

\begin{definition}[Hidden Subgroup Problem] Let $\mathrm{G}$ be a group and $f$ a function\footnote{The function $f$ in the hidden subgroup problem is said to be separating cosets of $\mathrm{H}$ as $f$ is constant on a each coset and different on different cosets.} from $\mathrm{G}$ to a set $\mathrm{X}$. We know that  $f(g_{0})=f(g_{1})$ if and only if $g_{0}\mathrm{H}=g_{1}\mathrm{H}$ for some subgroup $\mathrm{H}$. The problem is, given $f$ find a generating set for the unknown subgroup $\mathrm{H}$.\end{definition} 

We can now reduce the hidden shift problem with functions $f_{0}$ and $f_{1}$ defined above on the group $\mathrm{G}=\rm{\rm{\rm{GL}}}_{k} (\mathbb{F}_{2}) \times \mathrm{S}_{n}$ to the hidden subgroup problem over $(\mathrm{G}\times \mathrm{G})\rtimes \mathbb{Z}_{2}$~\cite[Section 2.2]{Dinh}. The hidden subgroup in this case is 
\begin{equation}\label{eqnK}
\mathrm{K}=(((\mathrm{H}_{0},s^{-1}\mathrm{H}_{0}s),0) \cup ((\mathrm{H}_{0}s,s^{-1}\mathrm{H}_{0}),1)),
\end{equation} where $\mathrm{H}_{0}=\lbrace (A,P)\in \rm{\rm{GL}}_{k} (\mathbb{F}_{2}) \times S_{n} : A^{-1}MP=M  \rbrace$ and $s$ is a shift from $f_{0}$ to $f_{1}$.

In short, the scrambler-permutation problem is one of the key ways to attack a \N cryptosystem. This problem can be formulated as a hidden shift problem which further can be reduced to a hidden subgroup problem. So we can attack \N cryptosystems by trying to solve a hidden subgroup problem over $(\mathrm{G}\times \mathrm{G})\rtimes \mathbb{Z}_{2}$ where
$\mathrm{G} =\rm{\rm{GL}}_{k} (\mathbb{F}_{2}) \times S_{n}$.

\subsection{Successful Quantum Attacks}
\noindent In the previous section we saw that solving the hidden subgroup problem as a standard way to attack a \N cryptosystem. An interesting question is, when is the hidden subgroup problem hard to solve? This way we can ensure the security of a \N cryptosystem against known quantum attacks.

We briefly sketch some thoughts behind effectiveness of QFS. The algorithm of QFS in a general scenario and its use for solving the hidden subgroup problem is
 very well explained by Grigni \etal~\cite{Vazirani}. Arguments particular to \N cryptosystems and corresponding hidden subgroup problem are described in
  details by Dinh \etal~\cite[Section 3]{Dinh}. The standard model of QFS yields a probability distribution as a function of the hidden subgroup. The basic
   idea behind \emph{indistinguisability} of two subgroups $\mathrm{H}_{1}$ and $\mathrm{H}_{2}$ with probability distributions $\mathsf{P}_{\mathrm{H}_{1}}$
    and $\mathsf{P}_{\mathrm{H}_{2}}$ is that  $\mathsf{P}_{\mathrm{H}_{1}}$ and $\mathsf{P}_{\mathrm{H}_{2}}$ are \emph{very close}. For the purpose of
     defining closeness we need to define a metric on the space of probability distributions. In this case, the metric chosen is the total variation distance between two distributions. This follows from the work of Kempe and Shalev~\cite{Kempe}. Later Dinh \etal~\cite{Dinh} used the $\mathcal{L}_1$ distance to
      define distinguishibility. Furthermore, the probability distribution for the hidden subgroup problem with the identity subgroup gives us the uniform distribution~\cite[Section 3.2]{Dinh}. When that is the case, QFS will not give us much information to solve the hidden subgroup problem. Kempe and Shalev~\cite{Kempe} provided a necessary condition to distinguish a subgroup of $\mathrm{S}_{n}$ from the trivial subgroup $\langle e \rangle$.
       Later Dinh \etal~\cite{Dinh} extended this result while keeping the group relevant to a \N cryptosystem in mind. Their result can be viewed as a study of the hidden subgroup problem for the  group $\mathrm{G}=(\rm{\rm{GL}}_{k}(\mathbb{F}_{2})\times \mathrm{S}_{n})^{2} \rtimes \mathbb{Z}_{2}$ which is the
        group for \N cryptosystems. They demonstrated a case when the hidden subgroup $\mathrm{H}$ can not be distinguished from either its conjugate subgroups $g\mathrm{H}g^{-1}$ or the trivial subgroup $\langle e \rangle$ and proved a general result on a sufficient condition for
         indistinguishibility~\cite[Theorem 4]{Dinh}. We use their result to prove that the proposed Niederreiter cryptosystem is quantum-secure.

First  note that weak Fourier sampling gives the same distributions for all the conjugate subgroups, \ie, $\mathsf{P}_{\mathrm{H}}$ is the same as $\mathsf{P}_{g\mathrm{H}g^{-1}}$. Hence weak Fourier sampling can not differentiate a subgroup from its conjugate subgroup. Thus it suffices to look at strong Fourier sampling. Dinh \etal~\cite{Dinh}, inspired by the work of Kempe and Shalev~\cite{Kempe} defines distinguishability of a subgroup $\mathrm{H}$ by strong Fourier sampling.

\begin{definition}[Distinguishability of a subgroup by strong QFS] 
We define distinguishability of a subgroup $\mathrm{H}$ of a group $\mathrm{G}$, denote it by $\mathrm{D}_\mathrm{H}$, to be the expectation of the squared $\mathcal{L}_1$ distance between $\mathsf{P}_{g\mathrm{H}g^{-1}}$ and the uniform distribution, where $g\in\mathrm{G}$. In other words,
\[\mathrm{D}_{\mathrm{H}} := \mathbf{E}_{\rho ,g} \left[\| \mathsf{P}_{g\mathrm{H}g^{-1}} (\cdot| \rho )-\mathsf{P}_{\langle e \rangle} (\cdot| \rho )\|^2_1\right].\]
A subgroup $\mathrm{H}$ is called \emph{indistinguishable} by strong Fourier sampling if $\mathrm{D}_{\mathrm{H}} \leq log^{- \omega(1)} \vert \mathrm{G} \vert$. The $\rho$ above belongs to the set of irreducible complex representations of the group $\mathrm{G}$.
\end{definition} 

Note that if a subgroup $\mathrm{H}$ is indistinguishable according to this definition then by Markov's inequality, for all $c>0$, $\Vert \mathsf{P}_{g\mathrm{H}g^{-1}} (\cdot| \rho ) - \mathsf{P}_{\lbrace e \rbrace} (\cdot| \rho )\Vert _{t.v.} \leq log^{-c} \vert \mathrm{G} \vert $ which is analogous to the definition provided by Kempe and Shalev~\cite{Kempe} for indistinguishability of a subgroup by weak Fourier sampling.
 
We now define the minimal degree of a permutation group, the automorphism group of a matrix (as defined by Dinh~\cite[Section 4.2]{Dinh}) and recall the definition of $T_M$ for a $k\times n$ matrix $M$.

\begin{definition}[Automorphism Group] The automorphism group of $M$ is defined as
$\rm{Aut}(M)=\lbrace P\in \mathrm{S}_{n}$ such that  there exists $ A \in \rm{\rm{\rm{GL}}}_{k}(\mathbb{F}_{2}) ,\ AMP=M\rbrace $.
 \end{definition}
\begin{definition}[Minimal Degree] The minimal degree of a permutation group $\mathrm{G} \leqslant \mathrm{S}_{n}$ acting on set of $n$ symbols is defined to be minimum number of elements moved by a non-identity element of the group $\mathrm{G}$.
\end{definition}
\begin{definition}Consider a $k\times n$ matrix $M$, we define $T_M$ for the matrix $M= \left[\,I_{k} | M^{*} \right]$ as 
$T_M=\lbrace \mathcal{P}_{1} \in \mathrm{S}_{k}$ such that there exists $\mathcal{P}_{2}\in \mathrm{S}_{n-k}$ with $\mathcal{P}_{1} M^{*} \mathcal{P}_{2} =M^{*}\rbrace$.
\end{definition}
\noindent We will use the following theorem which we state for the convinience of the reader.
\begin{theorem}[Dinh \etal~{\cite[Theorem 4]{Dinh}}]\label{thm1}
 Assume $q^{k^2} \leqslant n^{an}$ for some constant $0 < a < 1/4$. Let m be the minimal degree of the automorphism group Aut(M). Then for sufficiently large n, the subgroup K, $D_{K} \leqslant O(\vert K \vert ^ 2 e^{-\delta m} ) $, where $\delta > 0$ is a constant.
\end{theorem}
\section{Proposed Cryptosystem}\label{variant}
\noindent In this section we explain the proposed Niederreiter cryptosystem and establish its quantum security against the hidden subgroup attack. 

Recall that our variant of the Niederreiter cryptosystem consists of a parity check matrix $\mathcal{H}$ defined as follows:
The matrix $\mathcal{H}$ is an array of circulants in the systematic form\footnote{A systematic matrix is a matrix whose first block is the identity.}, that is, $\mathcal{H}=\left[\,C_{0}=I\,|\,C_{1}\,|\,C_{2}\,|\,\ldots\,|\,C_{m-1} \,\right]$ where each $C_{i}$ is a circulant matrix of size $p$ (a prime) over $\mathbb{F}_{2^{l}}$. For simplicity let us denote $\left[\,C_{1}\,|\,C_{2}\,|\,\ldots\,|\,C_{m-1}\, \right]$ as $C$ so that $\mathcal{H}= \left[\,I\,|\,C \,\right]$. Recall the conditions of the parity check matrix $\mathcal{H}$ from Section~\ref{dd}. A parity check matrix $\mathcal{H}$ satisfying these conditions is easy to construct. We present a way to do so in the next section. Before that, we prove security of the proposed \N cryptosystem against quantum attacks. 
\subsection{Proof of Indistinguishability}
\noindent We prove indistinguishibility in a sequence of lemmas.
\begin{lemma} \label{ov1} Let $P \in \rm{Aut}(\mathcal{H})$ then $P= P_{1} \oplus P_{2}$ where $P_{1}$ is a block of size $p$ and $P_{2}$ is a block of size $(m-1)p$ and $P_1\oplus P_2$ is a block diagonal matrix of size $mp\times mp$ with the top block $P_1$ and the bottom block $P_2$.
\end{lemma}
\begin{proof} Let $P \in \rm{Aut}(\mathcal{H})$, from the definition of automorphism there is an $A$ such that $A\mathcal{H}P=\mathcal{H}$. 
This implies that
\[A \left[\,I\,|\,C\,\right] P = \left[\,A\,|\,AC \,\right]P= \left[\,I\,|\,C\, \right].\] 

As action of right multiplication by a permutation matrix permute columns, the above equality shows that $\left[\,A\,|\,AC\,\right]$ has same columns as $\left[\,I\,|\,C \,\right]$ possibly in different order. Now since every column of $C$ contains an entry from a proper extension of $\mathbb{F}_q$, no column of $A$ can be column of $C$. This forces $A$ to have same columns as $I$ and $AC$ to have same columns as that of $C$. Hence $P$ permutes first $p$ columns within themselves and last $(m-1)p$ columns in themselves. Hence every $P \in \rm{Aut}(\mathcal{H})$ can be broken into $P_{1} \oplus P_{2}$ so that $P_{1}$ acts on $I$ and $P_{2}$ acts on $C$. 
\end{proof} 
\noindent The next lemma is central to quantum-security. It gives us a way to move from $\mathcal{H}$ to $C$ by noting, the $P_1$ from the $P\in\rm{Aut}(\mathcal{H})$ is actually a member of $T_{\mathcal{H}}$.
\begin{lemma} \label{ov3} The cardinality of $\rm{Aut}(\mathcal{H})$ is the cardinality of the set $\left\{(P_{1},P_{2})\right\}$ that satisfy $P_{1}C P_{2}=C$ where $\mathcal{H}=\left[\,I\,|\,C\,\right]$ as defined earlier.
\end{lemma}
\begin{proof} 
The proof follows from the fact, if $P$ belongs to $\rm{Aut}(\mathcal{H})$, then $P=P_1\oplus P_2$. Then $A\left[\,I\,|\,C\,\right]P=\left[\,I\,|\,C\,\right]$ translates into $A\left[\,I\,|\,C\,\right](P_1\oplus P_2)=\left[\,I\,|\,C\,\right]$.  Keeping in mind the block diagonal nature of $P$, it follows that $\left[\,AIP_1\,|\,ACP_2\,\right]=\left[\,I\,|\,C\,\right]$. Then $A=P_1^{-1}$ and $P_1^{-1}CP_2=C$. This proves the lemma.
\end{proof}
\noindent The next lemma proves that for each $P_1$ there is atmost one $P_2$.
\begin{lemma}\label{ov4} Cardinality of the set $\lbrace\left(P_{1},P_{2} \right)$ that satisfy $P_{1}C {P}_{2}=C\rbrace$ equals $\vert T_{\mathcal{H}} \vert$.  
\end{lemma}
\begin{proof} Recall that $T_{\mathcal{H}} =  \lbrace P_{1}$ that satisfy $P_{1}CP_{2}=C \rbrace$. So it suffices to show that for every $P_{1}$ there is at most one $P_{2}$. Since no two columns of $C$ are identical, no two columns of $P_{1}C$ are identical. Hence, there is at most one way to re-order them to get back $C$. Thus for every $P_{1}$ there is at most one $P_{2}$.
\end{proof}
\begin{theorem}[Burnside~{\cite[Theorem 3.5B]{Dixon}}] \label{burni}
Let $G$ be a subgroup of $Sym(\mathbb{F}_{p})$ containing a $p$-cycle $\mu : \xi \mapsto \xi+1$. Then $G$ is either 2-transitive or $G \leq A\rm{\rm{GL}}_{1}(\mathbb{F}_{p})$ where $A\rm{\rm{GL}}_{1}(\mathbb{F}_{p})$ is the affine group over $p$.
 \end{theorem}
\noindent We prove a theorem on the size of the automorphism group of $\mathcal{H}$. 
\begin{theorem}
If $\mathcal{H}$ satisfies conditions I,II and III then $\vert \rm{Aut}(\mathcal{H}) \vert \leqslant p(p-1)$.
\end{theorem}
\begin{proof}
From Lemma~\ref{ov3} and Lemma~\ref{ov4}, the group $\rm{Aut}(\mathcal{H})$ has same size as $T_{\mathcal{H}}$. It is now easy to check that the circulant matrix $\mu$ with first row $[0,1,0,\ldots,0]$ of size $p$ belongs to $T_{\mathcal{H}}$. The corresponding $P_{2}$ will be a block diagonal $(m-1)p$ matrix with blocks of size $p$ and each consisting of $\mu^{-1}$. Now notice that the circulant matrix $\mu$ corresponds to the $p$-cycle $\xi\mapsto\xi+1$.
By our condition III, $T_{\mathcal{H}}$ is not 2-transitive. Now by Burnside's theorem $T_{\mathcal{H}} \leq A\rm{\rm{GL}}_{1}(\mathbb{F}_{p})$. Thus $\vert \rm{Aut}(\mathcal{H}) \vert \leqslant p (p-1)$.
\end{proof}
\noindent After this bound on the size of the automorphism group we move towards the minimal degree of the \rm{Aut}omorphism group.

\begin{lemma} The minimal degree of $\rm{Aut}(\mathcal{H})$ is bounded below by $p-1$.
\end{lemma}
\begin{proof}
Notice that any $P\in\rm{Aut}({\mathcal{H}})=P_1\oplus P_2$. By the twist, from $P\in\rm{Aut}(\mathcal{H})$ to $P_1^{-1}\in T_{\mathcal{H}}$, it is easy to see that $P_{1} \in A\rm{GL}_1(\mathbb{F}_{p})$. Then $P_{1}(x)=ax+b \pmod{p}$ for some $a,b\in\mathbb{F}_p$. If $P_{1}$ fixes two distinct points, then $a=1$ and $b=0$ is the only possible solution. This corresponds to the identity element and thus a non-identity element can not fix more that one point. So minimal degree of  $\rm{Aut}(\mathcal{H})$ is bounded below by $p-1$.
\end{proof}

\noindent We now prove the main theorem of this paper.
\begin{theorem}\label{mainthm}
Let $p$ be a prime and $m$ a positive integer bounded above by a polynomial in $p$, such that, $p\leq\frac{1}{4}m\left(\log{m}+\log{p}\right)$. Then the subgroup $K$ (Equation~\ref{eqnK}) defined above is indistinguishable.
\end{theorem}
\begin{proof}
We will use Theorem~\ref{thm1} in this proof. First note, the minimal degree is bounded below by   
$p-1$. Now it is well known that $|K|=2|H_0|^2$ and $|H_0|=|\rm{Aut}(\mathcal{H})|\times|\rm{Fix(\mathcal{H})}|$. We have shown that $|\rm{Aut}(\mathcal{H})|\leq p(p-1)$ and it is easy to see that $|\rm{Fix(\mathcal{H}})|=1$. Putting all these together, we see that $|K|^2e^{-\delta p}\leq 4p^8e^{-\delta p}$ for some positive constant $\delta$. However, from the bound on the size of $m$, it is obviously true that $4p^8e^{-\delta p}\leq \left(mp\log{(mp)}\right)^{-\omega(1)}$ for large enough $p$.

Now, if $p\leq am\left(\log{m}+\log{p}\right)$, then $p^2\leq amp\left(\log{m}+\log{p}\right)$ which gives $2^{p^2}\leq(mp)^{amp}$ for $0<a<\frac{1}{4}$. This satisfies the premise of Theorem~\ref{thm1} and hence $K$ is indistinguishable.
\end{proof}
\subsection{Constructing the Parity Check Matrix}
\noindent Now we address the last question about the proposed \N cryptosystem, how to construct a matrix $\mathcal{H}$ satisfying conditions I - V? Clearly, conditions I, II and III are trivial to set up and deserve no special attention. We suggest a particular way for construction of parity check matrix $\mathcal{H}$ so that condition IV is satisfied. It should be noted that there may be other ways to satisfy condition IV as well. 

Choose a pair of distinct elements $a,b\in \mathbb{F}_{2^l}$. Now construct $\mathcal{H}$ such that $C_{1}$ contains both $a$ and $b$ exactly once in each column and no other $C_{i}$ contains both $a$ and $b$. We restate this condition as our condition $\mathrm{IV}^\prime$. We could have replaced $C_{1}$ by any other $C_{i}$ for $i>1$ and the proof remains the same. For sake of simplicity we stick with $C_{1}$.
\begin{description}
\item[$\mathrm{IV}^\prime$] Two distinct elements $a,b\in \mathbb{F}_{2^l}$ occurs as entries of $C_{1}$ exactly once in each column and no other $C_i$ contain both $a$ and $b$.
\end{description}

\begin{lemma} If the matrix $\mathcal{H}$ satisfies $\mathrm{IV}^\prime$, it also satisfies $\mathrm{IV}$.
\end{lemma}
\begin{proof} Let $\mathcal{P}_{1} \in T_{\mathcal{H}}$. From $\mathcal{P}_1 C\mathcal{P}_2=C$ it follows
 that $C\mathcal{P}_2$ should have the same set of columns as $C$ but possibly in a different order. Let
  $\alpha$ denote the row of $a$ in the first column of $C_{1}$ and $\beta$ denote the row of $b$ in the
   same column. Now notice that every column in $C$ that contains both $a$ and $b$ contains them such that
    difference between rows of $a$ and $b$  is $\alpha - \beta$ mod $p$ where $p$ is the size of each
     circulant matrix. Now let $\sigma\in T_{\mathcal{H}}$ such that it sends $\beta$ to $\alpha$ and 
     $\alpha$ to $\beta$. It then follows from the fact that $p$ is a odd prime, $\alpha=\beta$ which
      contadicts our assumption. Hence, $T_{\mathcal{H}}$ is not 2 transitive. 
\end{proof}
Condition V can be easily satisfied using brute force and other means and this completes the construction of a parity check matrix $\mathcal{H}$ satisfying I - V and hence, a \textbf{Niederreiter cryptosystem that resists quantum Fourier sampling} is found.
\section{Advantages of the proposed cryptosystem}
\noindent One of the prime advantages of our proposed cryptosystem is quantum-security. Apart from that it also has high transmission rate which translated into high encryption rate. It is known that the current McEliece cryptosystem built on Goppa codes has transmission rate of about $0.52$. For a McEliece cryptosystem its rate is same as that of the transmission rate of the underlying code and is $\frac{k}{n}$. Niederreiter cryptosystems have a slightly different rates due to difference in their encryption algorithm. For a general cryptosystem its encryption rate or information rate can be defined as follows ~\cite[Chapter 6]{NN}:

Let $\mathcal{S}(C)$ denote possible number of palintexts and $\mathcal{T}(C)$ denote possible number of ciphertexts then information rate of the system is defined by 
\[\mathcal{R}(C) = \dfrac{\log\mathcal{S}(C)}{\log\mathcal{T}(C)}.\]
This information rate can be viewed as amount of information contained in one bit of ciphertext.  

Our proposed Niederreiter cryptosystem have good encryption rate (see Table~\ref{table}). This gives the proposed cryptosystem an edge over those constructed on classical Goppa codes or with GRS codes (generalized Reed-Solomon codes).

As discussed before another problem with McEliece and Niederreiter cryptosystems is large key size. Circulant matrices is a good choice when it comes to key-sizes. Matrices are 2-dimensional objects but circulant matrices behave like a 1-dimensional object as they can be described by their first row. Though this circulant structure is lost in public key due to the scrambler-permutation pair, the size of the key still remains smaller than the conventional Niederreiter cryptosystem. Our system is slightly better than original Niederreiter cryptosystem because of the less number of rows in the public key matrices. With $p=101$, this number is less than one-tenth of the original Niederreiter cryptosystem. Though there are two factors that increases the size of the matrix in our variant compared to original McEliece: one, our matrices have large number of columns and two, our system is based on extension field $\mathbb{F}_{2^{l}}$ which makes the effective size of the matrix $l$ times compared to McEliece which is based on $\mathbb{F}_{2}$. However, in most cases due to less number of rows the net result indicates that our system requires shorter keys than original McEliece. For instance, at 80-bit security with $p=101$ and $l=3$ our keys are almost half of the keys corresponding to original McEliece at same security level. At 256-bit security level with $p=211, t=40$ and $l=3$ the proposed cryptosystem has key size of about one-forth of the original McEliece.  
\begin{table*}[ht]
\caption{Parameters for the proposed \N cryptosystem.}
\centering
\small{
\label{table}
\begin{tabular}{|l|c|c|c|c|c|c|c|c|c|c|}
\hline\hline
Security & $p$ & $t$ & $m_C$ & $m_Q$ & $m$ & Probability & \multicolumn{2}{|l|}{Public Key Size}& Rate\\
\cline{8-9}
in bits & & & & & & of success & No.~rows & No.~cols & \\
\hline
\multirow{4}{*}{80-bits}                                                         & \multirow{2}{*}{101} & 15                        & 17 & 35  & 35 & $2^{-132}$                                                           & 101                                         & 3535 & 0.60 \Tstrut\\
                                                                                 &                      & 20                        & 9              & 35 & 35& $2^{-190}$                                                          & 101                                         & 3535 & 0.77 \Tstrut\\
                                                                                 & \multirow{2}{*}{211} & 35                        & 4 & 62 & 62 & $2^{-398}$                                                          & 211                                         & 13082 & 0.71 \Tstrut\\
                                                                                 &                      & 40                       & 3& 62 & 62 & $2^ {-465}$                                                          & 211                                         & 13082  & 0.80                                     \TBstrut\\
                                                                                 \hline
\multirow{4}{*}{100-bits}                                                       & \multirow{2}{*}{101} & 15                        & 40      & 35     & 40  & $2^ {-136}$                                                          & 101                                         & 4040 & 0.61\Tstrut\\
                                                                                 &                      & 20                        & 17 &  35 & 35&$2^{-190}$                                                          & 101                                         & 3535 & 0.77\Tstrut\\
                                                                                 & \multirow{2}{*}{211} & 35                       & 5 &  62 & 62 & $2^{-398}$                                                          & 211                                         & 13082 & 0.71 \Tstrut\\
                                                                                 &                      & 40                       & 5 & 62 & 62 & $2^{-465}$                                                          & 211                                         & 13082 & 0.80 \TBstrut\\
                                                                                 \hline
\multirow{4}{*}{120-bits}                                                      & \multirow{2}{*}{101} & 15                        & 95 & 35   & 95         & $2^{-171}$                                                          & 101                                         & 9595 & 0.67\Tstrut\\
                                                                                 &                      & 20                        & 32 & 35 & 35& $2^{-190}$                                                          & 101                                         & 3535 & 0.77 \Tstrut\\
                                                                                 & \multirow{2}{*}{211} & 35                        & 8& 62 & 62 & $2^{-398}$                                                          & 211                                         & 13082 & 0.71 \Tstrut\\
                                                                                 &                      & 40                       & 6& 62 & 62 & $2^{-465}$                                                          & 211                                         & 13082 & 0.80 \TBstrut\\
                                                                                 \hline                                                                                                  
                                                                                 
\multirow{2}{*}{256-bits}                                                      & \multirow{2}{*}{211} & 35                        & 98 & 62   & 98         & $2^{-443}$                                                          & 211                                        & 20678 & 0.75\Tstrut\\
                                                                                 &                      & 20                        & 55 & 62 & 62& $2^{-465}$                                                          & 211                                        & 13082 & 0.80 \Tstrut\\ \hline                                                                                                                                                                                             
\end{tabular}
}
\end{table*}

In Table~\ref{table}, we provide some parameters for the proposed \N cryptosystem and show in details the benefits of the proposed cryptosystem. There are two kind of attacks -- classical and quantum. For the classical we come out with a value of $m$ and call it $m_C$ and for quantum we call it $m_Q$. The maximim of these two is the $m$ that one should use for that said parameter. As explained earlier we use $p$ for the size of the matrix and $t$ as the error correcting capacity. We also provide success probability from classical attack, key size and the rate of the cryptosystem.
\section{Conclusion}
\noindent In this paper, we develop a Niederreiter cryptosystem using quasi-cyclic codes that is both classically and quantum secure against the current known attacks. In particular, we show that for the proposed cryptosystem the hidden subgroup problem from the  natural reduction of the corresponding scrambler-permutation problem is indistinguishable by quantum Fourier sampling. We also show that the proposed cryptosystem has high encryption rate and shorter keys compared to classical McEliece cryptosystems. One of the important problem that needs to be addressed is finding quasi-cyclic codes that satisfy the suggested parameter sizes. It would be interesting to see if the cryptosystem remains classically secure if we use other sparse keys. It is very clear that the system remains secure against quantum computers as the group structure for the system remains the same. This is important because it could reduce key sizes substantially.  
\bibliographystyle{plain}
{\small
\bibliography{paper}}

\begin{thebibliography}{10}

\bibitem{Aylaj}
Bouchaib Aylaj, Mostafa Belkasmi, Said Nouh, and Hamid Zouaki.
\newblock Good quasi-cyclic codes from circulant matrices concatenation using a
  heuristic model.
\newblock {\em International journal of advanced computer science and
  applications}, 7(9):63--68, 2016.

\bibitem{Baldi}
Marco Baldi, Marco Bodrato, and Franco Chiaraluce.
\newblock A new analysis of the {McEliece} cryptosystem based on {QC-LDPC}
  codes.
\newblock {\em Security and Cryptography for Networks}, pages 246--262, 2008.

\bibitem{berlekamp1978inherent}
Elwyn Berlekamp, Robert McEliece, and Henk Van~Tilborg.
\newblock On the inherent intractability of certain coding problems (corresp.).
\newblock {\em IEEE Transactions on Information Theory}, 24(3):384--386, 1978.

\bibitem{lange}
Daniel~J. Bernstein, Tanja Lange, and Christiane Peters.
\newblock Attacking and defending the {McEliece} cryptosystem.
\newblock In {\em Post-Quantum Cryptography. PQCrypto 2008}, pages 31--46.

\bibitem{blahut}
Richard~E. Blahut.
\newblock {\em Algebraic codes for data transmission}.
\newblock Cambridge University Press, 2003.

\bibitem{Dinh}
Hang Dinh, Cristopher Moore, and Alexander Russell.
\newblock {McEliece} and {Niederreiter} cryptosystems that resist quantum
  fourier sampling attacks.
\newblock volume 6841 of {\em LNCS}, 2011.
\newblock Crypto2011.

\bibitem{Dixon}
J.~D. Dixon and B.~Mortimer.
\newblock {\em Permutation Groups}.
\newblock Graduate Texts in Mathematics. Springer, New York, 1996.

\bibitem{Vazirani}
Michelangelo Grigni, Leonard Schulman, Monica Vazirani, and Umesh Vazirani.
\newblock Quantum mechanical algorithms for the nonabelian hidden subgroup
  problem.
\newblock In {\em Proceedings of the thirty-third annual ACM symposium on
  theory of computing}, pages 68--74. ACM, 2001.

\bibitem{Gulliver_thesis}
Thomas~A. Gulliver.
\newblock {\em Construction of quasi-cyclic codes}.
\newblock PhD thesis, University of Victoria, 1989.

\bibitem{hallgreen}
Sean Hallgrean, Alexander Russell, and Amnon Ta-Shma.
\newblock The hidden subgroup problem and quantum computation using group
  representation.
\newblock {\em SIAM Journal of Computation}, 32(4):916--934, 2003.

\bibitem{Hirotomo}
M.~Hirotomo, M.~Mohri, and M.~Morii.
\newblock A probabilistic computation method for the weight distribution of
  low-density parity-check codes.
\newblock In {\em International Symposium on Information Theory}, 2005.

\bibitem{th}
Upendra Kapshikar.
\newblock {McEliece}-type cryptosystems over quasi-cyclic codes.
\newblock Master's thesis, IISER Pune, 2018.
\newblock https://arxiv.org/abs/1805.09972.

\bibitem{Kempe}
Julia Kempe and Aner Shalev.
\newblock The hidden subgroup problem and permutation group theory.
\newblock In {\em Proceedings of the sixteenth annual ACM-SIAM symposium on
  discrete algorithms}, pages 1118--1125. Society for Industrial and Applied
  Mathematics, 2005.

\bibitem{Lee}
Pil~Joong Lee and Ernest~F Brickell.
\newblock An observation on the security of {McEliece's} public-key
  cryptosystem.
\newblock In {\em Eurocrypt 1988}, volume 330 of {\em LNCS}, pages 275--280.
  Springer, 1988.

\bibitem{Li}
Yuan~Xing Li, Robert~H Deng, and Xin~Mei Wang.
\newblock On the equivalence of {McEliece's} and {Niederreiter's} public-key
  cryptosystems.
\newblock {\em IEEE Transactions on Information Theory}, 40(1):271--273, 1994.

\bibitem{ME}
R.~J. McElice.
\newblock A public key cryptosystem based on algebraic coding theory.
\newblock Technical report, Communications system research centre, NASA,
  Jan-Feb 1978.

\bibitem{NN}
Harald Niederreiter and Chaoping Xing.
\newblock {\em Algebraic Geometry in Coding Theory and Cryptography}.
\newblock Princeton University Press, 2009.

\bibitem{stern1988method}
Jacques Stern.
\newblock A method for finding codewords of small weight.
\newblock In {\em International Colloquium on Coding Theory and Applications},
  pages 106--113. Springer, 1988.

\bibitem{zeh}
Alexander Zeh and San Ling.
\newblock Decoding of quasi-cyclic codes up to a new lower bound on the minimum
  distance.
\newblock In {\em 2014 IEEE International Symposium on Information Theory,
  Honolulu, HI, 2014}.

\end{thebibliography}
\Addresses
\end{document}